\pdfoutput=1
\documentclass{cccg21}

\usepackage[
  backend=biber,
  style=numeric,
]{biblatex}
\usepackage{mathtools}
\usepackage{amssymb}
\usepackage{suffix}

\newcommand{\Oh}{\mathcal{O}}
\newcommand{\oh}{o}
\newcommand{\Om}{\Omega}
\newcommand{\Th}{\Theta}

\newcommand{\bool}{{\set{0, 1}}}
\newcommand{\nat}{{\mathbb{N}_0}}
\newcommand{\real}{\mathbb{R}}

\DeclarePairedDelimiter{\paren}{(}{)}
\DeclarePairedDelimiter{\set}{\lbrace}{\rbrace}
\DeclarePairedDelimiter{\tuple}{\langle}{\rangle}
\DeclarePairedDelimiter{\abs}{|}{|}
\DeclarePairedDelimiter{\ceil}{\lceil}{\rceil}
\DeclarePairedDelimiter{\floor}{\lfloor}{\rfloor}
\newcommand{\setCond}[2]{\set{#1~\vert~#2}}
\WithSuffix\newcommand\setCond*[2]{\set*{#1~\middle\vert~#2}}
\DeclarePairedDelimiter{\intCO}{[}{)}

\newcommand{\access}{\mathbb{A}}
\newcommand{\prog}{\fun{R}}
\newcommand{\A}[1]{{\mathcal{A}_{#1}}}
\newcommand{\subProg}{\mathcal{S}}
\newcommand{\accessSub}{{\access_\subProg}}
\newcommand{\ASub}[1]{{\mathcal{A}^\subProg_{#1}}}

\newcommand{\fun}[1]{{\operatorname{\mathtt{#1}}}}
\newcommand{\funDef}[2]{\colon#1\rightarrow#2}

\newcommand{\hfrac}[2]{{#1 \mathbin{/} #2}}
\newcommand{\concat}{\mathbin{\|}}

\newcommand{\logeq}{\mathrel{\vcentcolon\Leftrightarrow}}

\newcommand{\ord}[1]{\mathrel{\leq_{#1}}}
\newcommand{\lea}{\ord{a}}
\newcommand{\leb}{\ord{b}}
\newcommand{\lep}{\ord{p}}

\newcommand{\markName}[1]{\relax\ifmmode\text{\underline{#1}}%
  \else\underline{#1}\fi}
\newcommand{\marker}[2][]{\fun{l}(\ifblank{#1}{%
  #2}{%
  #2, \markName{#1}})}

\newcommand{\dual}[1]{\overline{#1}}
\newcommand{\intersection}[2]{{p_{{#1}\times{#2}}}}

\undef\Pr
\newcommand{\Pr}[1]{{\operatorname{Pr}[#1]}}

\usepackage{algorithm}
\usepackage[noend]{algpseudocode}

\algdef{SE}[DOWHILE]{Do}{DoWhile}{\algorithmicdo}[1]{\algorithmicwhile\ #1}
\algrenewcommand{\textproc}[1]{$#1$}

\newtheorem{definition}{Definition}
\newtheorem{remark}{Remark}

\usepackage{array}
\newcolumntype{L}{>{$}l<{$}}
\newcolumntype{C}{>{$}c<{$}}
\newcolumntype{R}{>{$}r<{$}}

\usepackage{multirow}

\newcommand{\stremph}[1]{\textbf{\boldmath#1}}
\newcommand{\code}[1]{\texttt{#1}}

\usepackage{paralist}

\usepackage[
  hidelinks,
]{hyperref}
\usepackage[
  nameinlink,
  capitalize,
]{cleveref}
\newcommand{\bareRef}[2]{\hyperref[#2]{#1}}

\usepackage{csquotes}

\newcommand{\Cpp}{\texorpdfstring{%
  C\raisebox{0.5ex}{\tiny\textbf{++}}}{%
  C++}}

\newboolean{fullVersion}
\setboolean{fullVersion}{true}

\title{Oblivious Median Slope Selection}
\author{%
  Thore Thie{\ss}en\thanks{Westf\"alische Wilhelms-Universit\"at M\"unster,
    Dept.\ of Computer Science {\tt \{t.thiessen,jan.vahrenhold\}@uni-muenster.de}}
  \and
  Jan Vahrenhold\footnotemark[\value{footnote}]}

\begin{document}
  \thispagestyle{empty}
  \maketitle

  \begin{abstract}
    We study the median slope selection problem in the oblivious RAM
    model. In this model memory accesses have to be independent of
    the data processed, i.\,e., an adversary cannot use observed
    access patterns to derive additional information about the input.
    We show how to modify the randomized algorithm of
    \textcite{matousek_randomized_1991} to obtain an oblivious version
    with $\Oh(n \log^2 n)$ expected time for $n$~points in
    $\real^2$. This complexity matches a theoretical upper bound that
    can be obtained through general oblivious transformation. In addition, results
    from a proof-of-concept implementation show that our
    algorithm is also practically efficient.
  \end{abstract}

\section{Introduction}%
\label{sec:Introduction}

Data collected for statistical analysis is often sensitive in nature.
Given the increasing reliance on cloud-based solutions for data processing, there is a demand for data-processing techniques that provide privacy guarantees.
One such guarantee is \emph{obliviousness}, i.\,e., an algorithm's property to have externally observable runtime behavior that is independent of the data being processed.
Depending on the runtime behavior observed, oblivious algorithms can be used to perform privacy-preserving computations on externally stored data or mitigate side channel attacks on shared resources~\parencite{stefanov_oblivistore_2013,liu_memory_2013}.

In the oblivious RAM model of computation~\parencite{goldreich_towards_1987,goldreich_software_1996} algorithms need to be oblivious with respect to the memory access patterns; we refer to \emph{memory-access obliviousness} as \emph{obliviousness}.
In general this leads to an $\Om(\log m)$ overhead compared to RAM algorithms when operating on $m$ memory cells~\parencite{goldreich_software_1996,larsen_yes_2018,hofheinz_stronger_2019}.
A transformation approach matches this lower bound asymptotically~\parencite{asharov_optorama_2020}, but is known to result in prohibitively large constant runtime overhead.

The median slope, know as the \emph{Theil--Sen estimator}, is a linear point estimator that is robust against outliers~\parencite{sen_estimates_1968}.
The randomized algorithm of \textcite{matousek_randomized_1991} computes the median slope of $n$~points in $\real^2$ with expected runtime $\Oh(n \log n)$ and is fast in practice.
We derive an oblivious version of \citeauthor{matousek_randomized_1991}'s algorithm that is slower by a logarithmic factor --- matching the complexity obtainable through general transformation --- but still fast in practice.

\subsection{Median slope selection problem}%
\label{sec:ProblemDefinition}

Median slope selection is a special case of the general \emph{slope selection problem}:
Given a set of points $P$ in the plane, the slope selection problem for an integer $k$ is to select a line with $k$-th smallest slope
among all lines through points in $P$~\parencite{cole_optimal-time_1989}.
Formally, given a set of $n$ points $P \subset \real^2$ let $L \coloneqq \setCond{\set{p, q} \subseteq P}{p_x \not= q_x}$
be the set of all pairs of points from $P$ with distinct $x$-coordinates.
We use $\ell_{p q} \in L$ to denote the line through points $\set{p, q} \in L$.
No line in $L$ is vertical by definition,\footnote{%
  \textcite{cole_optimal-time_1989} allow the selection of vertical lines and thus points with identical $x$-coordinates, but we exclude these as the Theil--Sen estimator is defined for non-vertical lines only.} so the slope $m(\ell_{p q})$ is well-defined for all $\ell_{p q} \in L$.
Let $k$ be an integer with $k \in [\abs{L}] \coloneqq \set{0, \ldots, \abs{L} - 1}$.
The slope selection problem for $k$ then is to select points
$\set{p, q} \in L$ such that $\ell_{p q}$ has a $k$-th smallest slope
in $L$.

Unless noted otherwise and in line with \textcite{matousek_randomized_1991} our exposition assumes that the points $P$ are in general position:
All $x$-coordinates of points $\set{p, q} \subseteq P$ are distinct and all lines through different pairs of points have different slopes.
For simplicity we also assume that $\abs{L}$ is odd, so that the median slope can be determined by solving the slope selection problem for $k = \frac{\abs{L} - 1}{2}$.
In \cref{sec:NonGeneralPositions}, we discuss how to lift these restrictions.

\citeauthor{matousek_randomized_1991}'s algorithm approaches the slope selection problem by considering the dual \emph{intersection selection problem}~\parencite{matousek_randomized_1991}:
Each point $p = \tuple{p_x, p_y} \in P$ can be mapped to dual non-vertical line $\dual{p}\colon x \mapsto (p_x x - p_y)$ and vice versa.
Since we have $\dual{p}(x) = \dual{q}(x) = y \iff \dual{\tuple{x, y}} = \ell_{p q}$,
a point in the set $\dual{L}$ of (dual) intersection points
with $k$-th smallest $x$-coordinate is dual to a line in $L$ with $k$-th smallest slope~\parencite{matousek_randomized_1991,dillencourt_randomized_1992}.

We thus restrict ourselves to finding an intersection of dual lines $\dual{P}$ with $k$-th smallest $x$-coordinate.
By the above assumption regarding the general position of the points in $P$, the lines in $\dual{P}$ have distinct slopes and all intersection points have distinct $x$-coordinates.

\subsection{Oblivious RAM model}%
\label{sec:ObliviousRamModel}

We work in the \emph{oblivious RAM} (\emph{ORAM}) model~\parencite{goldreich_towards_1987,goldreich_software_1996}.
This model is concerned with what can be derived by an adversary observing the \emph{memory access patterns} during the execution of a program.
The general requirement is that memory accesses are (data-)\emph{oblivious}, i.\,e., that the adversary can learn nothing about the input (or output) from the memory access pattern.

In line with standard assumptions, we assume a probabilistic word RAM with word length $w$, a constant number of registers in the processing unit and access to $m \leq 2^w$ memory cells with $w$ bits each in the memory unit~\parencite{hofheinz_stronger_2019}.
The constant number of registers in the processing unit are called \emph{private memory} and do not have to be accessed in an oblivious manner.

Whether a given probabilistic RAM program $\prog$ operating on inputs $X$ is oblivious depends on the way memory is accessed.
Let
$  \access \coloneqq \set{\fun{read}, \fun{write}} \times [m]$
be the set of memory \emph{probes} observable by the adversary.
Each probe is identified by the memory operation and the access location $i \in [m]$.
The random variable
$  \A{\prog(x)}\funDef{\Omega}{\access^*}$
denotes the \emph{probe sequence} performed by $\prog$ for an input $x \in X$ where $\Omega \coloneqq \bool^{l \cdot w}$ is the set of possible random tape contents.
The program $\prog$ is \emph{secure} if no adversary, given inputs $x, x' \in X$ of equal length and a probe sequence $A \in \access^*$, can reliably decide whether $A$ was induced by $x$ or $x'$.
For a program with an output determined by the input this implies that no adversary can decide between given outputs~\parencite{farach-colton_bucket_2020}.

We operationalize obliviousness by restricting the definition of \textcite{chan_cache-oblivious_2018} to perfect security, determined programs, and perfect correctness.
The definition also generalizes the allowed dependence of the probe sequence on the length of the input to a general \emph{leakage};
the leakage determines what information the adversary may be able to derive from the memory access patterns.

\begin{definition}[Oblivious simulation]%
  \label{def:ObliviousnessDefinition}
  Let $f\funDef{X}{Y}$ be a computable function and let $\prog$ be a probabilistic RAM program.
  $\prog$ obliviously simulates $f$ with regard to leakage $\fun{leak}\funDef{X}{\bool^*}$ if $\prog$ is \stremph{correct}, i.\,e., for all inputs $x \in X$ the equality
$    \Pr{\prog(x) = f(x)} = 1$ holds,
  and if $\prog$ is \stremph{secure}, i.\,e., for all inputs $x, x' \in X$ with $\fun{leak}(x) = \fun{leak}(x')$ the equality
$    \sum_{A \in \access^*} \abs*{\Pr{\A{\prog(x)} = A} - \Pr{\A{\prog(x')} = A}} = 0$ holds.
\end{definition}

The composition of oblivious programs is also oblivious if the sub-procedures invoke each other in an oblivious manner\ifthenelse{\boolean{fullVersion}}{;
  see \bareRef{Appendix A}{sec:ObliviousComposability} for a more detailed discussion and proof of composability.%
}{.}
Here relaxing the leakage allows us to place fewer restrictions on sub-procedures while maintaining obliviousness of the complete program.

For the specific problem in this paper,
the algorithm is only allowed to leak the number of given lines, or, for subroutines, the length of each given input array.
We will prove the obliviousness of our algorithm by composability, so we will consider the obliviousness of sub-procedures individually.
In line with \cref{def:ObliviousnessDefinition} we will show the obliviousness of each procedure in relation to the input.
Since we only consider sub-procedures with determined result this implies the obliviousness in relation to the output.

\subsection{Related work}

There exists a breadth of research on the slope selection problem.
\textcite{cole_optimal-time_1989} prove a lower bound of $\Om(n \log n)$ for the general slope selection problem in the algebraic decision tree model that also holds in our setting\ifthenelse{\boolean{fullVersion}}{
  (see \bareRef{Appendix B}{sec:MedianSlopeSelectionLowerBound})%
}{}.
Both deterministic algorithms~\parencite{cole_optimal-time_1989,katz_optimal_1993,bronnimann_optimal_1998} and randomized~\parencite{matousek_randomized_1991,dillencourt_randomized_1992} algorithms have been proposed that achieve an $\Oh(n \log n)$ (expected) runtime.
The problem has also been considered in other models, see, e.\,g., in-place algorithms~\parencite{calamoneri_-place_2006}.

\textcite{asharov_optorama_2020} recently proposed an asymptotically optimal ORAM construction that matches the overhead factor of $\Om(\log m)$ per memory access.
This construction provides a general way to transform RAM programs into oblivious variants with no more than logarithmic overhead per memory operation.
Due to large constants this optimal oblivious transformation is not viable in practice, though practically efficient (yet asymptotically suboptimal) constructions are available, see, e.\,g., \emph{Path ORAM}~\parencite{stefanov_path_2013}.
Our algorithm matches the asymptotic runtime of an optimal transformation while maintaining practical efficiency and perfect security.

A different approach is the design of problem-specific algorithms without providing general program transformations.
Oblivious algorithms for fundamental problems have been considered, e.\,g., for sorting~\parencite{goodrich_randomized_2010,farach-colton_bucket_2020}, sampling~\parencite{sasy_oblivious_2019,shi_path_2020}, database joins~\parencite{agrawal_sovereign_2006,li_privacy_2008,krastnikov_efficient_2020}, and some geometric problems~\parencite{eppstein_privacy-preserving_2010}.
To the best of our knowledge neither the slope selection problem nor the related inversion counting problem have been considered in the oblivious setting before.

\section{A simple algorithm}%
\label{sec:Algorithm}

As mentioned above, our approach is to modify the randomized algorithm proposed by \textcite{matousek_randomized_1991}.
For this, we replace all non-trivial building blocks of the original algorithm --- most notably intersection counting and intersection sampling --- by oblivious counterparts.

\subsection{The original algorithm}%
\label{sec:OriginalAlgorithm}

\begin{algorithm*}
  \caption{Randomized intersection selection algorithm~\parencite{matousek_randomized_1991}.}%
  \label{alg:MatIntSelection}
  
\begin{algorithmic}[1]
  \Function{\fun{IntSelection}}{$\dual{P}, k$}%
    \Comment{$k \in [\fun{IntCount}(\dual{P}, -\infty, +\infty)]$}
  \State $n \gets \abs{\dual{P}}; N \gets \fun{IntCount}(\dual{P}, a, b)$%
  \Comment{Number of input lines and of remaining intersections}
  \State $a \gets -\infty; b \gets +\infty$
  \Do\label{alg:MatIntSelection:l:LoopStart}
    \State $j \gets \hfrac{n \cdot (k - \fun{IntCount}(\dual{P}, -\infty, a) + 1)}{N} - 1$%
      \Comment{Adjust $k$ relative to current boundaries}
    \State $j_a \gets \max\set{0, \floor{j - 3 \sqrt{n}}}; j_b \gets \min\set{n - 1, \ceil{j + 3 \sqrt{n}}}$
    \State $R \gets \fun{IntSample}(\dual{P}, a, b, n)$%
      \label{alg:MatIntSelection:l:Sampling} \Comment{Sample
        intersection points}
    \State $a' \gets \fun{Select}_x(R, j_a); b' \gets \fun{Select}_x(R, j_b)$%
      \label{alg:MatIntSelection:l:Selection}  \Comment{Select candidate
        boundaries for next iteration}
    \State $m_{a'} \gets \fun{IntCount}(\dual{P}, -\infty, a'); m_{b'} \gets \fun{IntCount}(\dual{P}, -\infty, b')$%
      \Comment{Count intersections left of $a'$ and left of $b'$}
    \If{$m_{a'} \leq k < m_{b'} \land m_{b'} - m_{a'} \leq \hfrac{11 N}{\sqrt{n}}$}%
      \label{alg:MatIntSelection:l:Check}
      \State $a, b \gets a', b'$ \Comment{Update boundaries}
      \State $N \gets m_{b'} - m_{a'}$ \Comment{Update remaining intersections}
    \EndIf
  \DoWhile{$N > n$}%
    \label{alg:MatIntSelection:l:LoopEnd}
  \State $R \gets \fun{IntEnumeration}(\dual{P}, a, b)$%
    \label{alg:MatIntSelection:l:Enumeration}%
    \Comment{Enumerate all remaining intersections}
  \State\Return $\fun{Select}_x(R, k - \fun{IntCount}(\dual{P},
  -\infty, a))$
  \Comment{Select correct intersection}
  \EndFunction
\end{algorithmic}

\end{algorithm*}

\Cref{alg:MatIntSelection} shows the original algorithm as described by \textcite{matousek_randomized_1991}.
In a nutshell the algorithm works by maintaining intersections $a$ and $b$ as lower and upper bounds for the intersection $p_k$ with $k$-th smallest $x$-coordinate to be identified.\footnote{
  Generalizing the description of the algorithm~\parencite{matousek_randomized_1991} we maintain the intersections $a, b$ instead of only their $x$-coordinates.}

In the main loop a \emph{randomized interpolating search} is performed, tightening the bounds $a$ and $b$ until only $N \in \Oh(n)$ intersections remain in between.
For this, a multiset $R$ of $n$ intersections is sampled from the remaining intersections (with replacement) in each iteration.
Then new bounds $a'$ and $b'$ are selected from $R$ based on the relative position of $p_k$ among the remaining intersections.
The check in \cref{alg:MatIntSelection:l:Check} ensures that $p_k$ lies within these new bounds and that the number of intersections has been sufficiently reduced.
\citeauthor{matousek_randomized_1991} proves that this check has a high probability to pass, implying that the number of remaining intersections is reduced by a factor of $\Om(\sqrt{n})$ in an expected constant number of iterations.
Thus only an expected constant number of loop iterations are required overall.
With only $N \in \Oh(n)$ intersections remaining, the solution is computed by enumeration and selection.

The only non-standard building blocks required for the algorithm are intersection counting, the sampling of $n$ intersections and the enumeration of intersections, all in a given range.
Due to the composability of oblivious programs the use of oblivious replacements in \cref{alg:MatIntSelection} leads to an oblivious algorithm; see \cref{sec:ObliviousnessAndRuntime}.

\subsection{Known oblivious building blocks}%
\label{sec:KnownObliviousPrimitives}

\paragraph{Sorting} In the ORAM model, $n$ elements can be sorted by a comparison-based algorithm in optimal $\Th(n \log n)$ time, e.\,g., using optimal sorting networks~\parencite{ajtai_on_1983}.
We refer to this building block as $\fun{Sort}$.

For the application in this paper we require a sorting algorithm which is fast in practice.
To this end we can use \emph{bucket oblivious sort}~\parencite{farach-colton_bucket_2020}:
The algorithm works by performing an oblivious random permutation step, followed by a comparison-based sorting step.
The random permutation ensures that the complete algorithm is oblivious, even if the sorting step is not.

Choosing suitable parameters ($Z \in \Th(\log n)$) we achieve a failure probability bounded by a constant~\parencite[Lemma~3.1]{farach-colton_bucket_2020}.
Since failure of the random permutation leaks nothing about the input, we can repeat this step until it succeeds.
Together with an optimal comparison-based sorting algorithm this results in an implementation for $\fun{Sort}$ that has an expected runtime of $\Oh(n \log n \log^2\log n)$ and is fast in practice.

\paragraph{Merging} The building block $\fun{Merge}(A, B)$ takes two individually sorted arrays $A$ and $B$ and sorts the concatenation $A \concat B$.
There is a lower bound of $\Om(n \log n)$ for merging in the indivisible oblivious RAM model.\footnote{
  \textcite{lin_can_2018} prove a lower bound of $\Om(n \log n)$ for stable partition in the indivisible oblivious RAM model that also applies to merging.
  This bound applies even when restricting the input to arrays of (nearly) equal size~\parencite{miltersen_asymptotic_1996}.}
\emph{Odd-even merge}~\parencite{batcher_sorting_1968,knuth_sorting_1973} is an optimal merge algorithm (in the indivisible oblivious RAM model) with a good performance in practice.

\paragraph{Selection} $\fun{Select}(A, k)$ denotes the selection of an element with rank $k$, i.e., a $k$-th smallest element, from an unordered array $A$.
An optimal algorithm in the RAM model is Blum's linear-time selection
algorithm. This problem can be solved by a near-linear oblivious algorithm~\parencite{lin_can_2019}, but current implementations suffer from high constant runtime factors due to the use of oblivious partitioning.
For practical efficiency, we realize selection by sorting the given array $A$.
Since for our application we may leak the index $k$, only one additional probe is required.
We thus have leakage $\fun{leak}\colon \tuple{A, k} \mapsto \tuple{\abs{A}, k}$.

\paragraph{Filtering} Filtering a field $A$ with a predicate $\operatorname{Pred}$ ($\fun{Filter}_{\operatorname{Pred}}(A)$) extracts a sorted sub-list $A'$ with all elements for which the predicate is true.
The elements $a \in A'$ are stable swapped to the front of $A$ and the number $\abs{A'}$ of such elements is returned.
Since filtering can be used to realize stable partitioning, the lower runtime bound of $\Om(n \log n)$ for inputs of length $n$ in the indivisible ORAM model~\parencite{lin_can_2019} applies.
This operation can be implemented with runtime $\Oh(n \log n)$ using oblivious routing networks~\parencite{goodrich_data-oblivious_2011}.

\paragraph{Appending} The building block $\fun{Append}(A, B, i, k)$ is given two fields $A$ and $B$ as well as two indices $i$ and $k$ and appends the first $k$ elements of $B$ to the first $i$ elements of $A$.
This ensures that $A'$ after the operation contains $A[0 : i] \concat B[0 : k]$ in the first $i + k$ positions.
All other positions may contain arbitrary elements.
This operation can also be implemented with runtime $\Oh(n \log n)$ by using oblivious routing networks.

\subsection{New oblivious building blocks}%
\label{sec:NewObliviousPrimitives}

\subsubsection{Inversion and intersection counting}%
\label{sec:InversionCounting}

The number of inversions in an array $A$ is defined as the number of pairs of indices $i, j \in [\abs{A}]$ with $A[i] > A[j]$ and $i < j$.
In the RAM model, an optimal comparison-based approach to determine
the number of inversions is a modified merge sort.
Our oblivious merge-based inversion counting $\fun{Inversions}$ generalizes this to an arbitrary merge algorithm (with indivisible keys).

As noted by \textcite{cole_optimal-time_1989}, inversion counting can be used to calculate the number of intersections of a set of lines in a given range $\intCO{a_x, b_x}$.
This is by ordering the lines according to the $y$-coordinates at $x = a_x$ ($\lea$) and counting inversions relative to the order at $x = b_x$ ($\leb$).
We use this to implement $\fun{IntCount}(\dual{P}, a, b)$ for determining the number of intersections of lines $\dual{P}$:

\begin{algorithm}
  \caption{Intersection counting.}%
  \label{alg:IntCount}
  
\begin{algorithmic}[1]
  \Function{\fun{IntCount}}{$\dual{P}, a, b$}
    \State $\dual{P}_a \gets \fun{Sort}_a(\dual{P})$%
      \Comment{Sort according to $\lea$}
    \State\Return $\fun{Inversions}_b(\dual{P}_a)$%
      \Comment{Count inversions}
  \EndFunction
\end{algorithmic}

\end{algorithm}

Given an array $A$ of elements (in our
case: lines sorted according to $\leq_a$), $\fun{Inversions}$
computes all inversions (in our case: corresponding to intersections in $\intCO{a_x, b_x}$) while at the same time sorting $A$. $\fun{Inversions}$ recursively computes all inversions
in the first half $A_{\mathrm{lo}}$ and in the second half $A_{\mathrm{hi}}$ of the
input. The inversions induced by lines from different halves, i.\,e.,
the number of pairs $\tuple{a, b} \in A_{\mathrm{lo}} \times A_{\mathrm{hi}}$ with $a < b$,
then is computed by $\fun{BiInversions}(A_{\mathrm{lo}}, A_{\mathrm{hi}})$ which leverages
that $A_{\mathrm{lo}}$ and $A_{\mathrm{hi}}$ may be assumed inductively to be sorted.

\begin{algorithm}
  \caption{Merge-based inversion counting.}%
  \label{alg:OblivInvCount}
  
\begin{algorithmic}[1]
  \Procedure{\fun{BiInversions}}{$A_{\mathrm{lo}}, A_{\mathrm{hi}}$}
    \State $\marker{e} \gets 0$, $e \in A_{\mathrm{lo}}$; $\marker{e} \gets 1$, $e \in A_{\mathrm{hi}}$%
      \Comment{Label}
    \State $A \gets \fun{Merge}(A_{\mathrm{lo}}, A_{\mathrm{hi}})$%
      \label{alg:OblivInvCount:l:Merge}%
      \Comment{Permute labels as well}
    \State $I \gets 0; c \gets 0$ \Comment{No.\ of inversions / counter}
    \For{$e \gets A[0], \ldots, A[\abs{A} - 1]$}%
      \label{alg:OblivInvCount:l:LoopStart}%
      \If{$\marker{e} = 0$}
        \State $I \gets I + c$ \Comment{Record inversions}
      \Else \Comment{$\marker{e} = 1$}
        \State $c \gets c + 1$ \Comment{Increase counter}
      \EndIf
    \EndFor%
      \label{alg:OblivInvCount:l:LoopEnd}
    \State\Return $I$
  \EndProcedure
\end{algorithmic}

\end{algorithm}

To do this obliviously, $\fun{BiInversions}$ labels the elements according to which half they come from, then merges the labeled elements, and finally uses these labels to simulate the standard RAM merging algorithm.
For this algorithm to work correctly, in general a stable merge algorithm is
required, which sorts elements from the first half before elements
from the second half if they are equal with regard to the order.  We can
drop this requirement since we only work on totally ordered inputs of unique
elements.

The correctness of inversion counting follows from the correctness of
$\fun{BiInversions}$. Independent of the particular merge algorithm
used, $\fun{BiInversions}$ is functionally equivalent to the merging
step of the RAM algorithm.  The runtime of $\fun{BiInversions}$ is
dominated by merging, thus $\fun{Inversions}$ runs in time
$\Oh(n \log^2 n)$.  As merging has a lower bound of $\Om(n \log n)$ in
the indivisible ORAM model and, even without assuming
indivisibility, no ORAM algorithm with
runtime $\oh(n \log n)$ is known, any divide-and-conquer approach
based on 2-way merges currently incurs a runtime of
$\Om(n \log^2 n)$.

Except for the invocation of $\fun{Merge}$, all operations in $\fun{BiInversions}$ can be realized obliviously by a constant number of linear scans over the elements $A \coloneqq A_{\mathrm{lo}} \concat A_{\mathrm{hi}}$ and their labels.
Since $\fun{Merge}$ is oblivious, the obliviousness of $\fun{BiInversions}$ follows from the composability of oblivious programs.
The obliviousness of $\fun{Inversions}$ again follows from composability.
Finally, since the input is divided depending only on the size of the input, $\fun{Inversions}$ and $\fun{IntCount}$ only leak the input size.

\paragraph{Defining a suitable order} Intuitively, the algorithm sorts the input (lines sorted according to $\leq_a$) according to $\leq_b$ while recording intersection points. At each such point, two lines adjacent in the underlying order exchange their position. In addition to handling boundary cases correctly, it is not immediately obvious how this approach can be modified to handle non-general positions, since there may be an arbitrary number of lines intersecting in a single point.

To be able to handle non-general positions obliviously, we do not explicitly use the $y$-coordinates to define $\lea$ and $\leb$.
Instead, we --- more generally --- order the lines by their intersection points in relation to a given intersection~$p$.
For this, we use $\intersection{i}{j} \coloneqq \ell_i \cap \ell_j$ to denote the intersection point of two lines $\ell_i \neq \ell_j$.

\begin{definition}%
  \label{def:LineOrder}
  Let $P_\times \coloneqq \dual{L} \cup \set{-\infty, +\infty}$
  be the set of all intersections formed by $\dual{P}$ with additional elements $-\infty$ and $+\infty$.
  Let also $\preceq$ be an order over $P_\times$ (with the corresponding strict order $\prec$).
  For each $p \in P_\times$, we define the binary relation $\lep$ over $\dual{P}$ as
  \begin{equation*}
    \ell_1 \lep \ell_2 \logeq \begin{cases}
      \top & \text{if $\ell_1 = \ell_2$} \\
      p \preceq \intersection{1}{2} & \text{if $m(\ell_1) > m(\ell_2)$} \\
      \intersection{1}{2} \prec p & \text{if $m(\ell_1) < m(\ell_2)$}
    \end{cases}
  \end{equation*}
\end{definition}

For lines in general position, this definition essentially captures the ordering by $y$-coordinate:
If the slope of $\ell_1$ is larger than the slope of $\ell_2$, $\ell_1$ lies below $\ell_2$ if their intersection point lies to the right of $p$; if the slope of $\ell_1$ is smaller than the slope of $\ell_2$, $\ell_1$ lies above $\ell_2$ if and only if their intersection point lies to the right.

\begin{lemma}[Correctness of $\fun{IntCount}$]%
  \label{lma:IntCountCorrectness}
  Let $P_\times$, $\preceq$, $\prec$, and $\lep$ be as defined above.
  If
  \begin{compactitem}
    \item[(a)] $\preceq$ is a total order over $P_\times$ with minimum $-\infty$ and maximum $+\infty$ and
    \item[(b)] $\lep$ is a total order over $\dual{P}$ for all $p \in P_\times$,
  \end{compactitem}
  then, given $a, b \in P_\times$ with $a \preceq b$,   $\fun{IntCount}$ determines the number of intersections $p \in \dual{L}$ with $a \preceq p \prec b$.
\end{lemma}

\begin{proof}
   $\fun{IntCount}$ sorts according to the order $\lea$ and then counts inversions according to the order $\leb$.
  The algorithm thus exactly counts the number of unique pairs $\set{\ell_1, \ell_2} \subseteq \dual{P}$ (assuming w.l.o.g.\ $m(\ell_1) > m(\ell_2)$) for which $(\ell_1 \lea \ell_2) \not= (\ell_1 \leb \ell_2)$.
  Since $\preceq$ is a total order and $a \preceq b$ this can only occur if $\ell_1 \lea \ell_2 \land \ell_1 \not\leb \ell_2$.
  Then $a \preceq \intersection{1}{2} \prec b$ follows directly from the definition of $\lep$, thus  $\fun{IntCount}$ counts exactly the number of intersections in the range $\intCO{a, b}$.
\end{proof}

Since we want to identify the intersection with median $x$-coordinate, the intersections need to be ordered primarily by their $x$-coordinate. If all intersection points have distinct $x$-coordinates --- which is the case for lines $\dual{P}$ in general position --- we have:

\begin{remark}
  Let $\dual{P}$ be in general position and $\preceq$ be defined as $p \preceq q \logeq p_x \leq q_x$
  for $p, q \in P_\times$ with special cases $-\infty \preceq p$ and $p \preceq +\infty$ for all $p \in P_\times$.
  Then both conditions in \cref{lma:IntCountCorrectness} are satisfied.
\end{remark}

We will prove this more generally in \cref{sec:IntersectionPointsIdenticalXCoordinates}.

The intersection point of two given lines can be determined in constant time, so $\lep$ can be evaluated in constant time as well.
As such the runtime of $\fun{IntCount}$ is dominated by $\fun{Inversions}$ and thus $\Oh(n \log^2 n)$ for $n$ given lines.
The method is oblivious by composability.

\subsubsection{Intersection sampling and enumeration}%
\label{sec:IntersectionSamplingEnumeration}

The last building blocks to consider are the independent sampling as well as the enumeration of intersection points from a given range $\intCO{a, b}$.
We need to avoid calculating all intersections explicitly, as this would result in a runtime of $\Oh(n^2)$.
Recall that sampling can be done efficiently in the RAM model by modifying the standard intersection counting algorithm:
First, a set $K$ of $k$ indices from the range $[\fun{IntCount}(\dual{P}, a, b)]$ are sampled and then the intersection count is computed while iterating over the generated indices, reporting the corresponding intersections on the fly~\parencite{matousek_randomized_1991}.

Unfortunately, this approach is not oblivious:
First, synchronized iterations (such as over $K$ and the set of intersections generated) are not oblivious in general as step widths depend on the data values encountered. Second, reporting an intersection on the fly leaks information about the lines inducing it.

We address these challenges in the following way. Just as we have done in $\fun{BiInversions}$, we simulate a synchronized traversal over arrays $A$ and $B$ by first sorting the (labeled) elements and then iterating over their concatenation $A \concat B$.
For each element, we decide in private memory how to process the element based on its label.

To avoid leaking information about the two lines inducing a single
intersection, we operate on batches producing partial results padded
to their maximum possible length where needed.  This way we do not
leak the number of samples from a specific sub-range of the input.

We combine sampling and enumerating into a single building block
$\fun{IntCollect}(\dual{P}, a, b, K)$; $K$ contains the indices
of the intersections to sample in ascending order.

\begin{algorithm}
  \caption{Enumerating specified intersections.}%
  \label{alg:OblivIntCollect}
  
\begin{algorithmic}[1]
  \Function{\fun{IntCollect}}{$\dual{P}, a, b, K$}%
      \Comment{$a \prec b$, $\abs{K} > 0$}
    \State $k' \gets 0$; $K' \gets \fun{array}[\abs{K}]$%
      \Comment{Intersection storage}
    \State $\dual{P}_a \gets \fun{Sort}_a(\dual{P})$%
      \Comment{Sort according to $\lea$}
    \State $I \gets 0$ \Comment{Intersection counter}
    \For{$l \gets 0, \ldots, \ceil{\log_2 \abs{\dual{P}}} - 1$}%
    \Comment{All layers}
      \State $\fun{DetermineLineIndices}_b(\dual{P}_a, I, l)$%
        \label{alg:OblivIntCollect:l:Indices}%
        \Comment{Upd.\ $I$}
      \State $X \gets \fun{MatchAgainstLines}(\dual{P}_a, K, l)$%
        \label{alg:OblivIntCollect:l:Matching}
      \State $\fun{StoreIntersections}(X, K', k')$%
        \label{alg:OblivIntCollect:l:Storing}%
        \Comment{Upd.\ $k'$}
    \EndFor
    \State\Return $K'$
  \EndFunction
\end{algorithmic}

\end{algorithm}

From a high-level perspective, the algorithm first sorts the input according to $\lea$ and then iteratively implements a bottom-up divide-and-conquer strategy:
As in the RAM algorithm sketched before, unique consecutive indices are (implicitly) assigned to all encountered intersection points.
Note that, as we randomly sample/enumerate intersections, we may assign indices to the intersections arbitrarily.
All lines are explicitly labeled with indices so that --- given the index for an intersection --- the lines inducing that intersection can easily be identified.

The intersection indices $K$ are then matched against the lines, determining the inducing lines of each intersection.
Finally, we store the pair of inducing lines as intersection in $K'$.
These three steps are repeated for each layer $l$ so that after processing all layers the inducing lines of all specified intersections are known.

We now discuss the routines called for each layer $l$.

\begin{table*}\centering
  
\setlength\tabcolsep{4pt}
\begin{tabular}{R @{\qquad} C C | C C | C C | C C c C C C C | C C C C}
  & \multicolumn{8}{c}{\emph{before merge}} & & \multicolumn{8}{c}{\emph{after merge}} \\
  \markName{i} & \multicolumn{4}{C |}{0} & \multicolumn{4}{C}{1} & \multirow{5}{*}{\qquad\ensuremath{\rightarrow}\qquad} & \multicolumn{4}{C |}{0} & \multicolumn{4}{C}{1} \\
  \markName{half} & 0 & 0 & 1 & 1 & 0 & 0 & 1 & 1 & & 0 & 1 & 1 & 0 & 1 & 1 & 0 & 0 \\
  e \in L & \ell_0 & \ell_6 & \ell_1 & \ell_5 & \ell_4 & \ell_7 & \ell_2 & \ell_3 & & \ell_0 & \ell_1 & \ell_5 & \ell_6 & \ell_2 & \ell_3 & \ell_4 & \ell_7 \\
  \markName{0-index} & & & & & & & & & & 3 & & & 3 & & & 5 & 7 \\
  \markName{1-index} & & & & & & & & & & 0 & 0 & 1 & 2 & 2 & 3 & 2 & 2 \\
\end{tabular}

  \caption{Labels assigned by $\fun{DetermineLineIndices}$ in layer $l
    = 1$ for an input of $8$ lines $\ell_0, \ldots, \ell_7$, numbered
    according to their $\leb$-order.
    In layer $0$, $I = 3$ inversions have been counted.
    Layer $1$ contains $6$ inversions.}%
  \label{tbl:DetermineLineIndicesExample}
\end{table*}

\begin{table*}\centering
  
\begin{tabular}{r c | C C C C C C | c}
  assigned index & & 3 & 4 & 5 & 6 & 7 & 8 & \\
  \hline
  intersection point & \multirow{4}{*}{} & \intersection{6}{1} & \intersection{6}{5} & \intersection{4}{2} & \intersection{4}{3} & \intersection{7}{2} & \intersection{7}{3} & \multirow{4}{*}{} \\
  \markName{i} & & 0 & 0 & 1 & 1 & 1 & 1 & \\
  \markName{0-index} (index of inducing 0-line) & & 3 & 3 & 5 & 5 & 7 & 7 & \\
  \markName{1-index} (index of inducing 1-line) & & 0 & 1 & 2 & 3 & 2 & 3 & \\
\end{tabular}

  \caption{Result of the merging step shown in \cref{tbl:DetermineLineIndicesExample}.
    Note that the indices are only assigned conceptually and the intersections are not computed explicitly.
    The assigned index is equal to $\markName{0-index} + (\markName{1-index} - \markName{i} \cdot 2^l)$.}%
  \label{tbl:IntersectionPointIndices}
\end{table*}

\paragraph{Assigning indices to lines}

The first sub-routine called for each layer $l$ is
$\fun{DetermineLineIndices}$. Building on the general ideas used in
\cref{alg:OblivInvCount}, it iterates over pairs of subarrays of $2^l$
lines each, updates the intersection counter $I$, and assigns to each
line in $\dual{P}$ four indices defined below that guide the oblivious
sampling.

\begin{compactitem}
\item The index \markName{i} of a line $\ell$ (or: $\marker[i]{\ell}$)
  denotes the pair of blocks (on the current layer) containing
  $\ell$. On each layer, we process only intersections of lines with
  the same index \markName{i}.
\item The index \markName{half}  of a line $\ell$ indicates whether  $\ell$ was stored in the
  first subarray $\dual{P}_{\mathrm{lo}}$ ($\marker[half]{\ell} = 0$, ``0-line'') or in
  the second subarray $\dual{P}_{\mathrm{hi}}$ ($\marker[half]{\ell} = 1$, ``1-line'').
   For each pair of subarrays, we process only intersections of lines with
   different indices \markName{half}.
\item For a 0-line $\ell_0$, the index \markName{0-index} is the
  offset of the first intersection induced by $\ell_0$.
  By construction all intersections induced by $\ell_0$ in this layer have consecutive indices.
  For a 1-line, \markName{0-index} stores the number of intersections
  counted thus far, i.\,e., all lines are sorted by their values of \markName{0-index} after merging.
\item For a 1-line $\ell_1$, \markName{1-index} is the offset among all 1-lines in this layer.
  For a 0-line $\ell_0$, \markName{1-index} stores the number of intersection points induced
  by $\ell_0$.
\end{compactitem}

The resulting algorithm is given as \cref{alg:ModifiedInversions}.
\Cref{tbl:DetermineLineIndicesExample} shows the labels assigned by \cref{alg:ModifiedInversions} in layer $l = 1$ when processing a sample input.
The labels correspond to the indices implicitly assigned to the intersection points shown in \cref{tbl:IntersectionPointIndices}.
The indices are assigned to the lines so that an intersection with index $i \in K$ is induced by a 0-line $\ell_0$ with next lower \markName{0-index} relative to $i$.
The \markName{1-index} of the inducing 1-line $\ell_1$ then is
\begin{equation*}
  \marker[1-index]{\ell_1} = \underbrace{i - \marker[0-index]{\ell_0}}_{\mathclap{\text{relative index of $\ell_1$ in the current pair of blocks}}}~+~\marker[i]{\ell_0} \cdot 2^l
\end{equation*}

Like $\fun{BiInversions}$ the runtime of $\fun{BiInversions}'_b$ is dominated by the call to $\fun{Merge}$ and thus $\Oh(s \log s)$ for sorted $\dual{P}_{\mathrm{lo}}, \dual{P}_{\mathrm{hi}}$ of size $s$.
This means that the runtime of $\fun{DetermineLineIndices}$ is in $\Oh\paren*{\frac{n}{s} \cdot s \log s} \subseteq \Oh(n \log n)$.
$\fun{BiInversions}'_b$ is oblivious like $\fun{BiInversions}$ is.
Since the main loop for $\fun{DetermineLineIndices}$ only depends on $n$ and $l$, as does the size of the input to $\fun{BiInversions}'_b$, the procedure is oblivious by composability with regard to leakage
$\fun{leak}\colon \tuple{\dual{P}_a, I, l} \mapsto \tuple{\abs{\dual{P}_a}, l}$.

\begin{algorithm}[t]
  \caption{Assigning indices to lines.}%
  \label{alg:ModifiedInversions}
  
\begin{algorithmic}[1]
  \Procedure{\fun{DetermineLineIndices}}{$\dual{P}_a, I, l$}
  \State $c^\ast \gets 0$ \Comment{Ctr.\ for 1-lines on level $l$}
    \For{$i \gets 0, \ldots, \ceil*{\frac{\abs{\dual{P}_a}}{2\cdot 2^l}}$}%
      \Comment{Pairs of subarrays}
      \State $\dual{P}_{\mathrm{lo}} \gets \dual{P}_a[2\cdot i \cdot 2^l: 2\cdot
      (i+1) \cdot 2^l -1]$
      \State $\dual{P}_{\mathrm{hi}} \gets \dual{P}_a[2\cdot (i+1) \cdot 2^l: 2\cdot
      (i+2) \cdot 2^l -1]$
      \State $\marker[i]{\ell} \gets i$ for all $\ell \in \dual{P}_{\mathrm{lo}} \concat \dual{P}_{\mathrm{hi}}$
      \State $\fun{BiInversions}'_b(\dual{P}_{\mathrm{lo}}, \dual{P}_{\mathrm{hi}}, I, c^\ast)$
    \EndFor
  \EndProcedure
  \vspace{.5em}
  \Procedure{\fun{BiInversions}'_b}{$\dual{P}_{\mathrm{lo}}, \dual{P}_{\mathrm{hi}}, I, c^\ast$}
    \State $\marker[half]{\ell} \gets 0$, $\ell \in \dual{P}_{\mathrm{lo}}$;
      $\marker[half]{\ell} \gets 1$, $\ell \in \dual{P}_{\mathrm{hi}}$
    \State $A \gets \fun{Merge}(\dual{P}_{\mathrm{lo}}, \dual{P}_{\mathrm{hi}})$%
      \label{alg:SampleLineMatching:l:Merge}%
      \Comment{Permute labels as well}
    \State $c \gets 0$%
      \Comment{Local counter for 1-lines}
    \For{$\ell \gets A[0], \ldots, A[\abs{A} - 1]$}%
      \State $\marker[0-index]{\ell} \gets I$
      \If{$\marker[half]{\ell} = 0$}
        \State $\marker[1-index]{\ell} \gets c$ \Comment{Number of int.\ for $\ell$}
        \State $I \gets I + c$ \Comment{Update intersection count}
        \Else
        \Comment{$\marker[half]{\ell} = 1$}
        \State $\marker[1-index]{\ell} \gets c^\ast$ \Comment{Record offset}
        \State $c \gets c + 1$ \Comment{Update local counter}
        \State $c^\ast\gets c^\ast+1$ \Comment{Update level counter}
      \EndIf
    \EndFor
  \EndProcedure
\end{algorithmic}

\end{algorithm}

\begin{algorithm*}
  \caption{Method for matching intersection indices against lines.}%
  \label{alg:SampleLineMatching}
  
\begin{algorithmic}[1]
  \Function{\fun{MatchAgainstLines}}{$\dual{P}_a, K, l$}%
      \Comment{Lines $\dual{P}_a$ already have the appropriate labels}
    \State $\marker[K]{i} \gets \top$ for all $i \in K$%
      \Comment{Mark intersection indices}
    \State $\marker[0-index]{i} \gets i + 0.5$ for all $i \in K$
    \State $X \gets \fun{Merge}_\markName{0-index}(\dual{P}_a, K)$%
      \label{alg:SampleLineMatching:l:Match0Start}%
      \Comment{Merge lines and intersection indices}
    \State $\ell_0 \gets \bot$; $\ell_1 \gets \bot$
      \Comment{$\marker[0-index]{\bot} \coloneqq \marker[1-index]{\bot} \coloneqq 0$}
    \For{$e \gets X[0], \ldots, X[\abs{X} - 1]$}%
      \Comment{Iterate over lines and indices, ignoring 1-lines}
      \If{$\neg\marker[K]{e} \land \marker[half]{e} = 0 \land \marker[1-index]{e} > 0$}%
        \Comment{Found 0-line inducing intersections}
        \State $\ell_0 \gets e$
      \ElsIf{$\marker[K]{e} \land \marker[0-index]{e} < \marker[0-index]{\ell_0} + \marker[1-index]{\ell_0}$}%
        \label{alg:SampleLineMatching:l:Match0Check}%
        \Comment{Found intersection index}
        \State $\marker[0-line]{e} \gets \ell_0$%
          \Comment{Mark index with inducing 0-line}
        \State $\marker[1-index]{e} \gets \marker[i]{\ell_0} \cdot 2^l + \marker[0-index]{e} - \marker[0-index]{\ell_0}$%
          \Comment{Calculate offset of the inducing 1-line}
      \EndIf
    \EndFor%
      \label{alg:SampleLineMatching:l:Match0End}
    \State $\fun{Sort}_{\markName{1-index}}(X)$%
      \label{alg:SampleLineMatching:l:Match1Start}
    \For{$e \gets X[0], \ldots, X[\abs{X} - 1]$}%
      \Comment{Iterate over lines and indices, ignoring 0-lines}
      \If{$\neg\marker[K]{e} \land \marker[half]{e} = 1$}%
        \Comment{Found 1-line}
        \State $\ell_1 \gets e$
      \ElsIf{$\marker[K]{e} \land \marker[1-index]{e} = \marker[1-index]{\ell_1} + 0.5$}%
        \label{alg:SampleLineMatching:l:Match1Check}%
        \Comment{Found intersection index}
        \State $\marker[1-line]{e} \gets \ell_1$%
          \Comment{Mark index with inducing 1-line}
      \EndIf
    \EndFor%
      \label{alg:SampleLineMatching:l:Match1End}
    \State\Return $X$
  \EndFunction
\end{algorithmic}

\end{algorithm*}

\paragraph{Matching lines and indices}

The second sub-routine, $\fun{MatchAgainstLines}$ (\cref{alg:SampleLineMatching}), pairs the lines inducing intersection points encountered in this layer that correspond to indices in $K$.

First, the indices are matched against the 0-lines.
This is done by assigning each index $i \in K$ the \markName{0-index} $i + 0.5$ and then merging them with the lines (by \markName{0-index}).
When iterating over the merged sequence $X$, the 0-line inducing an intersection from this layer is exactly the last 0-line encountered before the index (that induces at least one intersection).
Each index $i \in K$ is labeled with the inducing 0-line $\ell_0$ as
\markName{0-line} and with the index of the corresponding 1-line as
\markName{1-index}; the \markName{1-index} can be determined from the indices assigned to $\ell_0$.

Similarly, the indices are matched against the 1-lines by sorting the array $X$ of lines and indices according to the \markName{1-index}.
When iterating over the sorted sequence, the previous 1-line before each intersection index is the second line inducing the intersection.
Each $i \in K$ already assigned a \markName{0-line} can thus be labeled with the inducing 1-line $\ell_1$ as \markName{1-line}.

The runtime is dominated by the runtime for
merging and sorting and thus is in
$ \Oh((n + k) \log (n + k))$ for $k \coloneqq \abs{K}$ and
$n \coloneqq \abs{\dual{P}_a}$.  The algorithm is oblivious since, in
addition to merging and sorting, it only consists of linear scans over
the array $X$.  The input size for merging and sorting
is at most $n + k$.  Although not explicitly shown it is
trivial to implement the loop bodies obliviously with respect to both
memory access and memory trace-obliviousness.
By composability, $\fun{MatchAgainstLines}$ is oblivious with
regard to leakage
$\fun{leak}\colon \tuple{\dual{P}_a, K, l} \mapsto \tuple{\abs{\dual{P}_a}, \abs{K}}$.

\paragraph{Storing intersection}

The third subroutine called for each layer, $\fun{StoreIntersections}$ (\cref{alg:MatchedIntStorage}), stores the intersections (consisting of the pairs of lines matched in the previous step) in $K'$.
Exactly the indices with an assigned \markName{0-line} (and thus also \markName{1-line}) have been found in this layer.
For storing, the building block $\fun{Append}$ is used where $k'$ is the number
of indices already stored in $K'$.
$\fun{Append}$ is oblivious and thus does not leak the number of
intersections $k_\Delta$ from this layer.
The runtime of this last step is dominated by the filtering and appending steps and thus realizable with runtime
$\Oh(n \log n)$
where $n \coloneqq \abs{X}$.
The obliviousness follows from composability with regard to leakage
$ \fun{leak}\colon \tuple{X, K', k'} \mapsto \tuple{\abs{X}, \abs{K}}$.

\begin{algorithm}
  \caption{Storing the sampled indices}%
  \label{alg:MatchedIntStorage}
  
\begin{algorithmic}[1]
  \Procedure{\fun{StoreSampledIntersections}}{$X, K', k'$}
    \State $k_\Delta \gets \fun{Filter}_{\markName{K} \land \markName{0-line}}(X)$%
      \Comment{Matched indices}
    \State $\fun{Append}(K', X, k', k_\Delta)$%
      \Comment{Append (pairs of) lines}
    \State $k' \gets k' + k_\Delta$
  \EndProcedure
\end{algorithmic}

\end{algorithm}

\paragraph{Runtime and obliviousness}

Let $n \coloneqq \abs{\dual{P}}$ be the number of lines and $k \coloneqq \abs{K}$.
The runtime of $\fun{IntCollect}$ is dominated by the main loop.
This results in a total runtime of
$\Oh(\log n (n + k) \log (n + k)) \stackrel{k \in \Oh(n)}{=} \Oh(n \log^2 n)$.
The number of iterations and the sequence of values for $l$ only depends on $n$ and sub-routines only leak $n$, $k$, $n + k$, or $l$.
Thus, $\fun{IntSample}$ is oblivious by composability with regard to leakage
$\fun{leak}\colon \tuple{\dual{P}, a, b, K} \mapsto \tuple{\abs{\dual{P}}, \abs{K}}$.

\subsection{Analysis}%
\label{sec:ObliviousnessAndRuntime}

Since our implementation of \citeauthor{matousek_randomized_1991}'s
algorithm replaces only the building blocks used internally, the
correctness and runtime properties follow from the respective analyses
of the building blocks. We thus have:
\begin{lemma}[Correctness and runtime]%
  \label{lma:CorrectnessAndRuntime}
  Let $\fun{IntSelection}$ be \cref{alg:MatIntSelection} instantiated
  with the oblivious building blocks described above.  Then, given a set
  $\dual{P}$ of $n$ lines in general position and an integer
  $k \in \left[\binom{n}{2}\right]$, $\fun{IntSelection}(\dual{P}, k)$
  determines the intersection with $k$-th smallest $x$-coordinate in
  expected $\Oh(n \log^2 n)$ time.
\end{lemma}

We now turn our attention to the analysis of the proposed algorithm's obliviousness.
Since oblivious programs are composable, we can prove the security by considering the leakage of each oblivious building block.

\begin{lemma}[Obliviousness]%
  \label{lma:Obliviousness}
  Let $\dual{P}$ be a set of $n$ lines in general position such that
  $\binom{\abs{\dual{P}}}{2}$ is odd. If \cref{alg:MatIntSelection}
  is instantiated with the oblivious building blocks described above,
  $\fun{MedianSelection}(\dual{P}) \coloneqq
  \fun{IntSelection}(\dual{P}, k)$ with
  $k \coloneqq \frac{\binom{\abs{\dual{P}}}{2} - 1}{2}$ obliviously
  realizes the median intersection selection with respect to leakage
  $\fun{leak}(\dual{P}) \coloneqq \abs{\dual{P}}$.
\end{lemma}

\begin{proof}
  For the proof, we need to show both the correctness and the security of the algorithm for the specified inputs.
  The requirements above imply that $k$ is an integer, so correctness follows from \cref{lma:CorrectnessAndRuntime}.
  It remains to show the security.

  The oblivious algorithm directly uses the building blocks
  $\subProg \coloneqq \tuple{\fun{Sort}, \fun{Select}, \fun{IntCount}, \fun{IntCollect}}$
  The building block $\fun{IntCollect}$ is used to realize $\fun{IntSample}(\dual{P}, a, b, k)$ by first determining the number of inversions $i \coloneqq \fun{IntCount}(\dual{P}, a, b)$ in range $\intCO{a, b}$, independently sampling $k$ random indices $K \in {[I]}^k$, sorting the indices $K$ and calling $\fun{IntCollect}(\dual{P}, a, b, K)$.
  Similarly $\fun{IntCollect}$ is used to realize $\fun{IntEnumeration}(\dual{P}, a, b)$ by initializing an array $K \coloneqq \tuple{0, \ldots, i - 1}$ and calling $\fun{IntCollect}$.
  All building blocks are oblivious, with $\fun{Select}$ additionally leaking the rank of the selected element, $\fun{IntSample}$ leaking the number of samples via the size of $K$ and $\fun{IntEnumeration}$ leaking the number of intersections in the given range, also via the size of $K$.
  The arithmetic expressions and assignments operate on a constant
  number of memory cells and are trivially oblivious.

  We first examine the values of $n$, $k$, $N$ and $N' \coloneqq \fun{IntCount}(\dual{P}, -\infty, a)$ throughout the execution of the algorithm.
  The value of $n$ remains constant and $k$ is fixed relative to $n$, so we consider the sequence
$    B = \tuple{\tuple{N_0, N'_0}, \tuple{N_1, N'_1}, \ldots, \tuple{N_m, N'_m}}$
  where $N_i, N'_i$ are the values for $N, N'$ after the $i$-th iteration of the main loop for a total of $m$ loop iterations.
  In each iteration of the main loop, $n$ intersections $R$ are chosen uniformly at random from the range $\intCO{a, b}$.
  Since $\dual{P}$ is in general position, the intersections of distinct pairs of lines are distinct and all intersections are totally ordered.
  This implies that the random distribution of $\fun{IntCount}(\dual{P}, -\infty, c)$ for an intersection $c$ with fixed rank in $R$ only depends on $n$, $N$ and $N'$.
  Both $j_a$ and $j_b$ are fixed relative to $n$, $N$ and $N'$, so the random distribution of the next values for $N$ and $N'$ is solely determined by $n$ and the previous values.
  Since initially $N_0 = \binom{n}{2}$ and $N'_0 = 0$ and the sequence ends with $N_m \leq n$, the random distribution of the complete sequence $B$ is solely determined by $n$.

  It can easily be seen that each sequence $B$ of values for $N, N'$ determines the sequence $A$ of memory probes and sub-procedure invocations.
  This implies that any sequence $A$ is equally likely for inputs of the same size and thus that $\fun{MedianSelection}$ is secure by composability.
\end{proof}

\section{Non-general positions}%
\label{sec:NonGeneralPositions}

For simplicity of exposition, we assumed so far that the lines
$\dual{P}$ are in general position, i.\,e., that all intersection points
of two lines in $\dual{P}$ have distinct $x$-coordinates and that all
lines in $\dual{P}$ have distinct slopes. We also assumed that the
number of intersection points is odd, so that the
median intersection point selection problem can always be solved by
one call to a general intersection point selection algorithm; this
latter assumption can be removed by computing both the element with
rank $k_1 = \floor*{\frac{N - 1}{2}}$ and with rank
$k_2 = \ceil*{\frac{N - 1}{2}}$ (for $N = \binom{\abs{\dual{P}}}{2}$) and returning their mean if there is
an even number of intersections~\parencite{sen_estimates_1968}.
Since $k_1$ and $k_2$ differ by one at most by one,
both intersections can be computed simultaneously with no significant
impact on the runtime.

In RAM algorithms, degenerate configurations are a
nuisance, but often can be handled by generic approaches~\parencite[e.\,g.][]{edelsbrunner_simulation_1990,schirra_precision_1996,yap_geometric_1990}.
For our proposed algorithm, we must take care that these approaches do not affect the
obliviousness. In particular, the runtime of
the algorithm must not depend on the number of intersection points
with identical $x$-coordinates; this rules out the problem-specific
technique described by \textcite{dillencourt_randomized_1992} to
explicitly handle non-general position.

Regarding arithmetic precision, we note that the only arithmetic
computation performed on the input values is the calculation of the
$x$-coordinate of an intersection point.  Thus, recall we are working in the word
RAM model, for fixed-point input values with $b$ bits
of precision the use of $2 (b + 1)$ bits of precision suffices to
perform all arithmetic computations exactly.

\paragraph{Parallel lines}%
\label{sec:ParallelLines}
For technical reasons, we first discuss how to deal with inputs in which lines are parallel,
i.\,e., for which we cannot assume distinctness of slopes.

Earlier on, we noted that our algorithm is allowed to leak the values
of $N$ and $k$.\footnote{Assuming the leakage of $k$ allows us to
  treat the original algorithm of
  \citeauthor{matousek_randomized_1991} as a black box.  The author
  proves an expected lower bound on the reduction of $N$ per loop
  iteration which is independent of $k$.  This does not necessarily imply that the
  exact reduction of $N$ is in fact independent of $k$.}  This means
that we cannot introduce data-dependency of these values and this, in
turn, implies that (a) pairs of parallel lines cannot simply be
excluded and that (b) $k$ cannot be adjusted based on the number of
pairs of parallel lines.

We address this using a problem-specific, controlled version of the
symbolic perturbation of \textcite{edelsbrunner_simulation_1990,yap_geometric_1990}.  We perturb the lines $\dual{P}$ in
such a way that each pair $\set{\ell_1, \ell_2}$ of lines intersects
in a single intersection point $\intersection{1}{2}$.
Let $V$ be the set of intersections induced by lines that were
parallel previous to the perturbation. We ensure that $V$ is
partitioned into $V=V_- \cup V_+$ such that $V_-$ and $V_+$ are (nearly)
equally sized and each $v \in V_-$ has a $x$-coordinate less and each
$v' \in V_+$: By equally distributing these \enquote{virtual} intersections
to the left and to the right of all \enquote{real} intersections we
maintain data-independent values of $N = \binom{n}{2}$ and $k = \frac{N - 1}{2}$.

To realize this (symbolic) perturbation, we follow
\textcite{edelsbrunner_simulation_1990} and introduce an
infinitesimally small value $\varepsilon > 0$. We then identify each
line $\ell\colon x \mapsto m(\ell) \cdot x + b(\ell)$ with the
perturbed line $\ell'\colon x \mapsto m'(\ell) \cdot x + b'(\ell)$
where
$m(\ell') \coloneqq m(\ell) + s_\ell \cdot \#_\ell \cdot
\varepsilon^2$,
$b(\ell') \coloneqq b(\ell) + \#_\ell \cdot \varepsilon$,
$s_\ell \in \set{-1, +1}$ is a factor to achieve the distribution into
$V_-$ and $V_+$, and $\#_\ell \in \nat$ is a unique index given to each
line with respect to the order of the line offsets, i.\,e.
$\forall \ell_1, \ell_2 \in \dual{P}\colon b(\ell_1) < b(\ell_2)
\implies \#_{\ell_1} < \#_{\ell_2}$.  We obtain the set $\dual{P}'$ of
perturbed lines.

Due to space constraints, we omit the details of how to compute
$\#_\ell$ and $s_\ell$ as well how to
avoid leaking the number of \enquote{virtual} intersections.

\paragraph{Intersections with identical \texorpdfstring{\boldmath$x$}{x}-coordinates}%
\label{sec:IntersectionPointsIdenticalXCoordinates}

To handle intersections $p, q \in \dual{L}$ with identical
$x$-coordinates without significantly affecting the runtime of the algorithm, we
establish a total order $\preceq$ over all intersections, so that the
lines $\dual{P}$ can be totally ordered relative to each intersection
$p$ as in \cref{def:LineOrder}.  For this, we characterize an
intersection by its inducing pair of lines and define an order based
on these lines' properties:

\begin{definition}%
  \label{def:IntersectionOrder}
  Let $P_\times \coloneqq \dual{L} \cup \set{-\infty, +\infty}$
  be the set of all intersections with additional elements $-\infty$ and $+\infty$.
  Let each $p \in \dual{L}$ be formed by lines $p_\uparrow$ and $p_\downarrow$ with $m(p_\uparrow) > m(p_\downarrow)$.
  We define a total order $\preceq$ over $P_\times$ via:
  \begin{equation*}
    p \preceq q \logeq \begin{cases}
      p_x < q_x & \text{if $p_x \not= q_x$} \\
      m(p_\uparrow) < m(q_\uparrow) & \text{else if $p_\uparrow \not= q_\uparrow$} \\
      m(p_\downarrow) \leq m(q_\downarrow) & \text{else}
    \end{cases}
  \end{equation*}
  for $p, q \in P_\times \setminus \set{-\infty, +\infty}$ and with special cases $-\infty \preceq p$ and $p \preceq +\infty$ for all $p \in P_\times$.
  Let $\prec$ denote the corresponding strict order over $P_\times$.
\end{definition}

By construction, $\preceq$ is a (lexicographic) total order.  This is
ensured by the fact that all slopes are distinct.  This order suffices
to construct the total order over the lines in $\dual{P}$. To show
this, we need the following lemma:

\begin{lemma}%
  \label{lma:ThreePointLemma}
  Let $\ell_1, \ell_2, \ell_3$ be non-vertical lines with $m(\ell_1) < m(\ell_2) < m(\ell_3)$.
  Of the three intersections induced by these lines, the intersection
  $\intersection{1}{3}$ of the two lines with extremal slopes is the
  median with respect to the order $\preceq$ defined above.
\end{lemma}

Assuming only the distinctness of slopes (which, as discussed
above, may be assumed w.l.o.g.), we have:

\begin{lemma}
  Let $P_\times$, $\preceq$, and $\prec$ be as in
  \cref{def:IntersectionOrder}.  For each $p \in P_\times$, we have a
  total order $\lep$ over $\dual{P}$:
  \begin{equation*}
    \ell_1 \lep \ell_2 \logeq \begin{cases}
      \top & \text{if $\ell_1 = \ell_2$} \\
      p \preceq \intersection{1}{2} & \text{if $m(\ell_1) > m(\ell_2)$} \\
      \intersection{1}{2} \prec p & \text{if $m(\ell_1) < m(\ell_2)$}
    \end{cases}
  \end{equation*}
\end{lemma}

With the above definition, we can impose a total order on the set of
lines irrespective of whether or not their intersection points'
$x$-coordinates are distinct.
Since the predicate $p \preceq q$ for intersections $p, q \in
P_\times$ can still be evaluated in constant time, the asymptotic
runtime of the algorithm remains unchanged.

\paragraph{Summary}

In conclusion, the two techniques sketched in this section generalize
the algorithm not only to inputs $\dual{P}$ with parallel lines, but
also to inputs with identical lines.  The algorithm is thus applicable
to arbitrary inputs.  Since we can achieve the desired (symbolic)
perturbation via pre-processing in $\Oh(n \log n)$ time for an input
of $n$ lines, our main theorem follows:

\begin{theorem}[Main result]
  There exists a RAM program that obliviously
  realizes the median intersection selection in expected
  $\Oh(n \log^2 n)$ time for $n$ non-vertical lines
  inducing at least one intersection.
\end{theorem}

\section{Implementation and evaluation}%
\label{sec:Evaluation}

We developed a prototype of our oblivious algorithm in \Cpp.\footnote{%
  \url{http://go.wwu.de/ms6fz}
}
The goal of the implementation is to show that the algorithm is easily implementable and to provide an estimate of the algorithm's performance.
For this we also implemented the baseline algorithm~\parencite{matousek_randomized_1991}.

\paragraph{Limitations}
The primary limitation is that our prototype only accesses arrays of non-constant size in an oblivious manner.
Code fragments such as inner loops and methods accessing only a
constant number of memory cells do not necessarily probe memory obliviously.
Even though it is conceptually trivial to transform those code fragments to achieve \enquote{full} obliviousness,
we note that --- without publicly available libraries providing
low-level primitives for implementations of oblivious algorithms ---
the obliviousness eventually might depend on the compiler and platform used.

We believe that our implementation still provides a good estimate of the performance of a \enquote{fully} oblivious implementation:
The loops in our runtime-intensive primitives are all linear scans over arrays.
As such \enquote{fully} oblivious loop bodies will not introduce a large overhead since they will likely not introduce cache misses.
Also our oblivious primitives can also be implemented largely without
data-dependent branches, thus potentially eliminating branch mispredictions.

The second main limitation is that we do not implement the handling of parallel lines (as described in \cref{sec:ParallelLines}).
This would require an additional pre-processing step as well as extending both the slope and the offset with a symbolic perturbation.
As mentioned above this would result in a low constant factor overhead in both runtime and memory space usage.
Since this applies to both the oblivious and non-oblivious algorithm
this has no direct implication for the performance evaluation below,
although there might be a more efficient way to handle identical
slopes in the non-oblivious case.

Finally our implementation resorts to a suboptimal, but
easy-to-implement oblivious sorting primitive with $\Oh(n \log^2 n)$
and thus has an expected $\Oh(n \log^3 n)$ runtime in the oblivious
setting. This leads to an additional $\Oh(\log n)$ overhead in runtime as
compared to our non-oblivious implementation and thus underestimates
the performance of the proposed algorithm.

\begin{figure}[t]
  \resizebox{\linewidth}{!}{\includegraphics{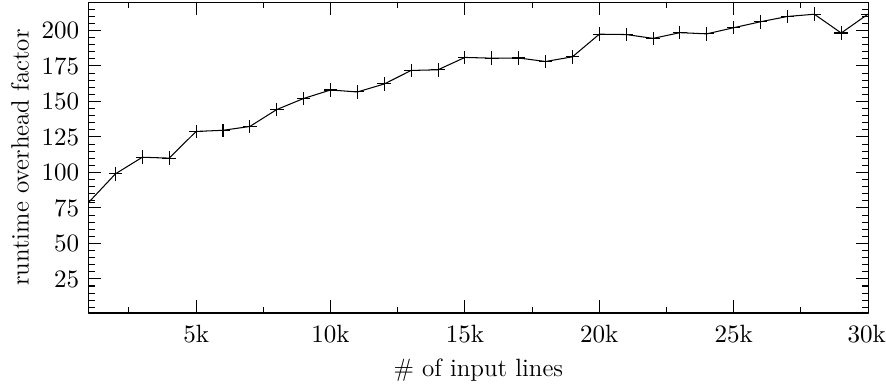}}
  \caption{Runtime overhead factor (averaged over 10 random inputs) of the oblivious algorithm compared to the baseline algorithm with non-oblivious primitives.}%
  \label{fig:RuntimeOverheadEvaluation}
\end{figure}

\paragraph{Performance}%
\label{sec:PerformanceEvaluation}

We used \code{libbenchmark}\footnote{%
  \url{https://github.com/google/benchmark} } to measure the runtime
for inputs ranging from 1,000 to 30,000 lines.  The input consists of
shuffled sets of lines with non-uniformly increasing slope and a
random offset, both represented by 64-bit integers.  For all our
experiments and independent of $n$, we fixed an interval
$[m_\text{min}, m_\text{max}]$ and an interval
$[b_\text{min}, b_\text{max}]$.  To generate a set of $n$ random
lines, we then set
$r \coloneqq \hfrac{(m_\text{max} - m_\text{min})}{n}$ and constructed
each line $\ell_i = \tuple{m_i, b_i}$ in turn by independently
sampling a random slope $m_i$ from
$m_\text{min} + i \cdot r \leq m_i < m_\text{min} + (i + 1) \cdot r$
(thus ensuring both spread and distinctness of slopes) and a random offset $b_i$ from
$b_\text{min} \leq b_i \leq b_\text{max}$.  We then permuted the
resulting set of lines using \code{std::ranges::shuffle}.

The performance evaluation results are shown in \cref{fig:RuntimeOverheadEvaluation}.
For inputs of 10,000--30,000 random lines our algorithm is about
150--210 times slower than the baseline algorithm.
While this is a significant slowdown, we remind the reader
of both the logarithmic overhead incurred by choosing a suboptimal
sorting algorithm and the fact that the baseline algorithm does not offer
any obliviousness.
The runtime was less than 10 seconds for all evaluated input sizes.

All experiments were performed on a Dell XPS~7390 with an Intel i7--10510U CPU and 16~GiB RAM running
Ubuntu~20.04.

\begin{figure}[t]
  \resizebox{\linewidth}{!}{\includegraphics{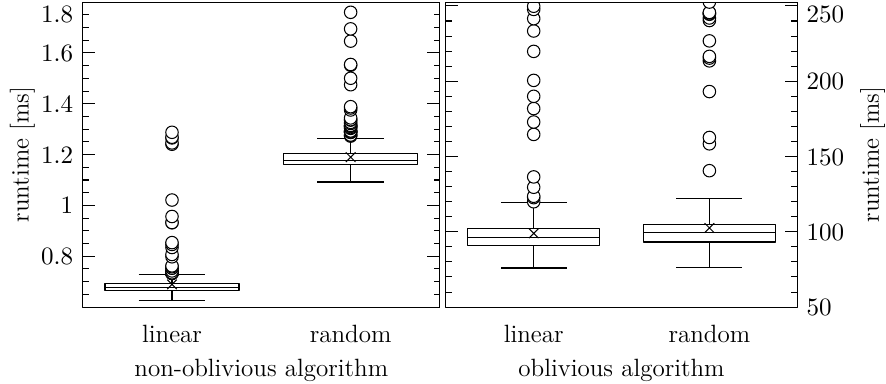}}
  \caption{Runtime distribution over $500$ runs of the oblivious and non-oblivious algorithms for $n = 1000$ lines.
    Left (in each subfigure): Data for a fixed, sorted input of lines intersecting in a single point.
    Right (in each subfigure): Data for a shuffled input of lines with non-uniformly increasing slopes and a random offset.}%
  \label{fig:RuntimeDifferenceEvaluation}
\end{figure}

\paragraph{Obliviousness}%
\label{sec:ObliviousnessEvaluation}

We assessed the obliviousness of our implementation of the building blocks by tracing memory accesses as part of unit testing.
For this, we abstracted the memory sections as arrays of fixed but dynamic size.
We assigned a fingerprint to each sequence of reads and writes by hashing both the memory operation and the access location.
Since all building blocks used by the main algorithm are deterministic, we asserted their obliviousness by comparing fingerprints for different inputs with identical leakage.

Additionally, we evaluated the runtime of both our oblivious algorithm and the baseline algorithm when applied to two inputs of different characteristics.
For this we compared the random lines described above with a sorted set of lines $\ell_i = \tuple{i, -i}$, intersecting in the single point $p = \tuple{1, 0}$.
The baseline algorithm showed significantly different runtimes for different inputs
(\cref{fig:RuntimeDifferenceEvaluation}), making it
abundantly clear that even without statistical analyses an
adversary can distinguish these different kinds of input from the runtime
alone. In contrast, there was only slight variation in the runtime of
our proposed algorithm which we attribute to the presence of code
processing constant-sized subproblems in a (currently) non-oblivious
manner.

\section{Conclusion}%
\label{sec:Conclusion}

We presented a modification of \citeauthor{matousek_randomized_1991}'s randomized algorithm~\parencite{matousek_randomized_1991} for obliviously determining the median slope for a given set of $n$ points.
We also showed how to generalize the algorithm to arbitrary inputs --- allowing both collinear points and multiple points with identical $x$-coordinate --- while maintaining obliviousness.
Our modified algorithm has an expected $\Oh(n \log^2 n)$ runtime, matching the general oblivious transformation bound of the original algorithm.
We provide a proof-of-concept of the oblivious algorithm in \Cpp{}, showing that the algorithm indeed can be implemented and has a runtime that make its application viable in practice.

  \setcounter{secnumdepth}{0}

  \section{References}
  {\small
  \printbibliography[heading=none]}

  \clearpage

\ifthenelse{\boolean{fullVersion}}{}{ }

\section{Appendices}

\subsection{A --- Composability of oblivious programs}%
\label{sec:ObliviousComposability}

Here we prove the composability of oblivious programs, i.\,e., that an oblivious program additionally invoking other oblivious programs as sub-procedures remains oblivious.
This is an adoption of the argument by \textcite{asharov_optorama_2018} to our definition of obliviousness.

Let $\subProg \coloneqq \tuple{\prog_0, \ldots, \prog_k}$ be probabilistic RAM programs with $k \geq 1$ such that $\prog_1, \ldots, \prog_k$ are oblivious.
We will first analyze the obliviousness of $\prog_0$ in the $\subProg$-hybrid RAM model:
Program $\prog_0$ may invoke any program from $\subProg$ as sub-procedure.
For the invocation of $\prog_i$ we assume that $\prog_0$ copies the input $x_i$ for $\prog_i$ to a new location $p$ in memory and executes a special machine instruction $\fun{invoke}(\prog_i, p)$ (only available in the $\subProg$-hybrid model).
The instruction immediately changes the partial memory beginning at location $p$ as if program $\prog_i$ were executed on the memory with offset $p$.
Since $\prog_i$ is executed as part of the instruction, any memory probes performed by $\prog_i$ are not part of the probe sequence of $\prog_0$ in the $\subProg$-hybrid model.
To ensure that the execution of $\prog_i$ does not interfere with the memory of $\prog_0$ a location $p$ after all used memory locations must be selected.
After the invocation $\prog_0$ can read the result computed by $\prog_i$ from memory.

In the $\subProg$-hybrid model we augment the probe sequence for $\prog_0$ with the sub-procedure invocations.
Similarly to the access locations for memory probes we identify each invocation by the invoked program and the leakage for the respective input:

\begin{definition}[Augmented probe sequence]
  Let $\access$ be as defined in \cref{sec:ObliviousRamModel}.
  Let $\subProg$ be as defined above and for each $\prog_i$ with $1 \leq i \leq k$ let $\fun{leak}_i\funDef{X_i}{\bool^*}$ be the leakage.
  We define the set of probes visible to the adversary during the execution of $\prog_0$ in the $\subProg$-hybrid model as
  $\accessSub \coloneqq \access \cup \set{\fun{invoke}} \times \subProg \times \bool^*$.

  The random variable
  $\ASub{\prog_0(x)}\funDef{\Omega}{\accessSub^*}$
  is defined to denote the hybrid sequence of probes and sub-procedure invocations by $\prog_0$ for input $x$.
  Specifically, for each memory probe $\fun{probe} \in \set{\fun{read}, \fun{write}}$ at location $i \in [N]$ the sequence contains an entry $\tuple{\fun{probe}, i}$.
  The sequence does not include memory probes performed by sub-procedures.
  For each invocation of sub-procedure $\prog_i \in \subProg$ with input $x_i$ the sequence contains an entry $\tuple{\fun{invoke}, \prog_i, \fun{leak}_i(x_i)}$.
\end{definition}

Note that according to the definition above the offset of the partial memory $p$ is not visible to the adversary.
This is a simplification which is justified by the fact that $\prog_0$ can always choose $p$ to be directly after the largest memory location written to before.
This guarantees that no memory contents are overwritten and implies that the adversary, given any probe sequence $A \in \accessSub^*$, can reconstruct $p$ for all sub-procedure invocations.

Invocations in the plain model can be realized by executing $\prog_i$ directly instead of the instruction $\fun{invoke}(\prog_i, p)$.
The offset of the partial memory $p$ can be held in a single special register and applied to every memory probe by a simple modification of $\prog_i$.
The register contents can be temporarily stored in memory and recovered after the execution of the sub-procedure.
In the plain model memory probes performed by $\prog_i$ are contained in the probe sequence.
The goal now is to show that obliviousness of $\prog_0$ in the $\subProg$-hybrid model implies obliviousness of the composed program $\prog_0$ in the plain RAM model:

\begin{lemma}[Composability of oblivious programs]
  Let $f_0, \ldots, f_k$ with $f_i\funDef{X_i}{Y_i}$ be computable functions, $\prog_0, \ldots, \prog_k$ randomized RAM programs and $\fun{leak}_0, \ldots, \fun{leak}_k$ with $\fun{leak}_i\funDef{X_i}{\bool^*}$ leakages.
  $\prog_0$ obliviously simulates $f_0$ with regard to leakage $\fun{leak}_0$ in the plain model if
  \begin{compactitem}
    \item[(a)] each $\prog_i$ for $1 \leq i \leq k$ obliviously simulates $f_i$ with respect to leakage $\fun{leak}_i$ in the plain model,
    \item[(b)] $\prog_0$ is correct in the $\subProg$-hybrid model, i.\,e., for all inputs $x \in X_0$ the equality
      $\Pr{\prog_0{(x)}^\subProg = f_0(x)} = 1$
      holds,
    \item[(c)] and $\prog_0$ is secure in the $\subProg$-hybrid model, i.\,e., for all inputs $x, x' \in X_0$ with $\fun{leak}_0(x) = \fun{leak}_0(x')$ the equality
      $\sum_{A \in \accessSub^*} \abs*{\Pr{\ASub{\prog_0(x)} = A} - \Pr{\ASub{\prog_0(x')} = A}} = 0$
      holds.
  \end{compactitem}
\end{lemma}

\begin{proof}
  To prove this we need to show that $\prog_0$ is both correct and secure in the plain model.

  The correctness immediately follows from the correctness of $\prog_0$ in the $\subProg$-hybrid model:
  Since the invocation in the plain and $\subProg$-models are functionally equivalent, termination with correct result in the $\subProg$-hybrid model implies termination with the correct result in the plain model.
  It thus only remains to show that $\prog_0$ is secure in the plain model, i.\,e., that any finite probe sequence $A$ is equally likely for any two inputs with identical leakage.

  We first consider the simple case where $\prog_0$ does not invoke itself.
  For this we fix any two inputs $x, x' \in X_0$ with $\fun{leak}_0(x) = \fun{leak}_0(x')$ and any finite probe sequence $A \in \accessSub^*$ for $\prog_0$ in the plain model.
  Consider any separation
  \begin{equation*}
    A = \tuple{A_0 \concat I_1 \concat A_1 \concat \ldots \concat I_n \concat A_n}
  \end{equation*}
  of $A$ where each $A_i$ consists of any number of memory probes performed by $\prog_0$ and each $I_i$ consists of the memory probes performed during the $i$-th invocation.
  Each separation of $A$ corresponds to any probe sequence
  \begin{equation*}
    A' = \tuple{A_0 \concat \tuple{\fun{invoke}, \prog_{j_1}, l_i} \concat \ldots \concat A_n} \in \accessSub^*
  \end{equation*}
  for $\prog_0$ in the $\subProg$-hybrid model.
  In $A'$ the memory probes $A_i$ by $\prog_0$ remain the same and sub-procedure $\prog_{j_i} \in \subProg \setminus \set{\prog_0}$ with some leakage $l_i$ performs memory probes $I_i$.

  We need to show the security
  \begin{equation*}
    \Pr{\A{\prog_0(x)} = A} = \Pr{\A{\prog_0(x')} = A}
  \end{equation*}
  in the plain model.
  Since the distribution of memory probes for the sub-procedures is independent of the memory probes $A_i$ by construction\footnote{
    This disregards the offset $p$ of the probe sequences $I_i$.
    The argument still holds since we can consider the distribution of the sequences $I_i$ shifted by $-p$ instead.
  }, it follows that the probability for $A$ with a specific $A'$ under input $y$ is exactly
  \begin{equation*}
    \Pr{\ASub{\prog_0(y)} = A'} \cdot \prod_{1 \leq i \leq n} \Pr{\A{\prog_{j_i}(y_i)} = I_i}
  \end{equation*}
  for some inputs $y_i$ to the sub-procedures with respective leakages $l_i$.
  From the security of $\prog_0$ in the $\subProg$-hybrid model and the security of each $\prog_{j_i}$ in the plain model it follows that this is the same for both $y \coloneqq x$ and $y \coloneqq x'$.
  Thus the probability for $A$ --- summing over all possible $A'$ --- is also the same and $\prog_0$ is secure in the plain model.

  For the second case we show that the lemma also holds when $\prog_0$ invokes itself recursively.
  We do so by induction over the depth $d$ of the recursion.
  The base case $d = 0$ corresponds to the case above when $\prog_0$ does not invoke itself.
  For $d > 0$ we consider finite probe sequences $A$ when each invocation may be either an invocation of sub-procedure $\prog_i \in \subProg \setminus \set{\prog_0}$ or an invocation of $\prog_0$ with a recursion depth of at most $d - 1$.
  By induction hypothesis for all invocations of $\prog_0$ with recursion depth of at most $d - 1$ the distribution of probe sequences is determined by the leakage of the input.
  Thus the same is true for any recursion depth $d \in \nat$.
  Since for any finite probe sequence $A$ the recursion depth of $\prog_0$ is also finite this proves the lemma.
\end{proof}

\subsection{B --- Lower bound}%
\label{sec:MedianSlopeSelectionLowerBound}

The lower bound of \textcite{cole_optimal-time_1989} for the general slope selection problem applies to problem definitions
\begin{compactitem}
  \item[(a)] allowing the selection of the smallest slope ($k = 0$) and
  \item[(b)] regarding slopes through points with equal $x$-coordinates as having a non-finite negative slope.
\end{compactitem}
For an arbitrary input $X = \tuple{x_1, \ldots, x_n} \in \real^n$ selecting the smallest slope through points $P = \set{\tuple{x_1, 1}, \ldots, \tuple{x_n, n}}$ yields a non-finite negative slope if and only if not all values in $X$ are distinct.
This proves, through reduction from the element uniqueness problem, a lower bound of $\Oh(n \log n)$ in the algebraic decision tree model.
\parencite{cole_optimal-time_1989}

This argument can be modified to prove a lower bound for the median slope selection problem also excluding non-finite slopes.

\begin{lemma}
  Let $P \subset \real^2$ be a set of $n$ points.
  Then, in general, determining the median slope through points in $P$ as defined in \cref{sec:ProblemDefinition} requires $\Om(n \log n)$ steps in the algebraic decision tree model.
\end{lemma}

\begin{proof}
  As done by \textcite{cole_optimal-time_1989} we reduce from the element uniqueness problem.
  Given an arbitrary input $X = \tuple{x_1, \ldots, x_n} \in \real^n$, we first transform $X$ to $X' = \tuple{x'_1, \ldots, x'_n}$ containing only positive values by subtracting less than the minimal value:
  \begin{equation*}
    x'_i \coloneqq x_i - \min_{1 \leq i \leq n} x_i + 1
  \end{equation*}
  For $1 \leq i < n$ we then map each value $x'_i$ to two points with equal $x$-coordinates:
  \begin{alignat*}{3}
    \tuple{i, x'_i} & \quad\text{and}\quad & \tuple{i, -x'_i}
  \end{alignat*}
  The last value $x'_n$ is mapped to two points with distinct $x$-coordinates:
  \begin{alignat*}{3}
    \tuple{n + 1, x'_n} & \quad\text{and}\quad & \tuple{n, -x'_n}
  \end{alignat*}
  Let $P \subset \real^2$ be the set of these $2 n$ points.

  Considering only lines $L$ through points in $P$ with positive $y$-coordinates, there exist some number $a \geq 0$ of lines with positive slope and some number $b \geq 0$ of lines with negative slope.
  Yet $L$ also contains a line with a slope of zero if and only if not all values in $X'$ (and thus also in $X$) are distinct.
  Looking at the lines $L'$ through points in $P$ with negative $y$-coordinates the same holds true, except because of the inverted $y$-coordinates there are exactly $a$ lines with negative and $b$ lines with positive slope in $L'$.
  $L'$ also contains the same number of lines with slope zero as $L$.

  It remains to consider the lines through points with different sign of the $y$-coordinate.
  Because all $x'_i$ are strictly positive and only looking at lines through points not mapped from the same $x'_i$, there are exactly
  \begin{equation*}
    \sum_{i = 0}^{n - 1} i = \frac{(n - 1) n}{2}
  \end{equation*}
  lines with positive and negative slope, respectively, and no lines with slope zero.
  Since for each $x'_i$ with $1 \leq i < n$ the two points are mapped to the same $x$-coordinate, lines through these points are not considered according to the problem definition.
  The exception is the line through the points $x'_n$ is mapped to, which has a positive slope.

  In summary, there are
  \begin{equation*}
    c \coloneqq a + b + \frac{(n - 1) n}{2}
  \end{equation*}
  lines with negative slope and $c + 1$ lines with positive slope as well as an even number of lines with slope zero.
  Thus, the median slope is zero if and only if a line with slope zero exists.
  This is the case if and only if not all values from $X$ are distinct.
\end{proof}

\end{document}